\documentclass[runningheads]{llncs}
\usepackage{graphicx}
\usepackage{amssymb,amsmath,amsfonts}
\usepackage{latexsym}
\usepackage{verbatim}
\usepackage{enumerate}
\usepackage{url}
\usepackage[linesnumbered,noend]{algorithm2e}
\usepackage{color}

\newcommand{\Spl}{\ensuremath{{\mathrm{Spl}}}} 

\renewcommand{\S}{\ensuremath{\mathcal{S}}} 

\newcommand{\Mred}{\ensuremath{\widetilde{M}}} 
\newcommand{\nred}{\ensuremath{\widetilde{n}}} 
\newcommand{\Sred}{\ensuremath{\widetilde{\S}}} 
\newcommand{\Hred}{\ensuremath{\widetilde{H}}} 
\newcommand{\Xred}{\ensuremath{\widetilde{X}}} 
\newcommand{\Ered}{\ensuremath{\widetilde{E}}} 
\newcommand{\cred}{\ensuremath{\widetilde{c}}} 

\newcommand{\QED}{\hfill$\Box$}

\SetKw{KwOr}{or}%
\SetKw{KwAnd}{and}%
\SetKw{KwContinue}{continue}%

\SetKwProg{Fn}{}{}{}
\SetKwFunction{FindNRC}{FindNRC}%

\begin{document}
\title{Checking Phylogenetic Decisiveness in Theory and in Practice}
%
%
\author{Ghazaleh Parvini\inst{1}
\and
Katherine Braught\inst{1}
\and
David Fern\'andez-Baca\inst{1}
}
\authorrunning{Parvini et al.}
%
\institute{Department of Computer Science, Iowa State University, Ames IA 50011, USA \\
\email{\{ghazaleh,kbraught,fernande\}@iastate.edu}}
\maketitle              
\begin{abstract}
Suppose we have a set $X$ consisting of $n$ taxa and we are given information from $k$ loci from which to construct a phylogeny for $X$.  Each locus offers information for only a fraction of the taxa.  The question is whether this data suffices to construct a reliable phylogeny.  The decisiveness problem expresses this question combinatorially. Although a precise  characterization of decisiveness is known, the complexity of the problem is open.  Here we relate decisiveness to a hypergraph coloring problem.  We use this idea to (1) obtain lower bounds on the amount of coverage needed to achieve decisiveness, (2) devise an exact algorithm for decisiveness, (3) develop problem reduction rules, and use them to obtain efficient algorithms for inputs with few loci, and (4) devise an integer linear programming formulation of the decisiveness problem, which allows us to analyze data sets that arise in practice.

\end{abstract}
%
%
%


\section{Introduction}  
\label{sec:incTaxCov}

Missing data poses a challenge to assembling phylogenetic trees.  The question we address here is how much data can one afford to miss without compromising accuracy.  We focus on data sets assembled by concatenating data from many (sometimes thousands) of loci \cite{JarvisEtAl_BirdPhyl2014,WickettEtAl_PhylotranscriptomicAnalysis2014,HinchliffPNAS2015}. Such data sets are used to construct phylogenetic trees by  either (i) combining the data from all the loci into a single  \emph{supermatrix}  that is then used as input to some standard phylogeny construction method (e.g., \cite{GuindonGascuel_SimpleML2003,stamatakis2014raxml}) or (ii)  taking phylogenetic trees computed separately for each locus and combining them into a single \emph{supertree} that summarizes their information  \cite{BinindaEmonds04,Scornavacca:2009,WarnowSupertrees2018}.
For various reasons, the \emph{coverage density} of concatenated datasets --- i.e., the ratio of the amount of available data to the maximum possible amount --- is often much less than $1$ \cite{SandersonMS2010}. 
Reference \cite{DobrinZwicklSanderson2018} examines a wide range of phylogenetic analyses using concatenated data sets, and reports coverage densities ranging from 0.06 to 0.98, with the majority being under 0.5.

Low coverage density can give rise to ambiguity \cite{Wilkinson_CopingMissingEntries1995,StamatakisAlachiotis_Time2010,SandersonTerracesScience2011}.  In supertree analyses, ambiguity manifests itself in multiple supertrees that are equivalent with respect to the method upon which they are based.  In super-matrix analyses, it is manifested in multiple topologically different, but co-optimal (in terms of parsimony or likelihood scores) trees.  Note that high coverage density does not, by itself, guarantee lack of ambiguity.  More important is the coverage pattern itself.  The question is whether one can identify conditions under which a given coverage pattern guarantees a unique solution.  Sanderson and Steel \cite{SteelSanderson2010,SandersonTerracesScience2011} have proposed a formal approach to studying this question, which we explain next.

A \emph{taxon coverage pattern} for a taxon set $X$ is a collection of sets $\S = \{Y_1, Y_2, \dots , Y_k\}$, where, for each $i \in \{1, 2, \dots , k\}$, $Y_i$ is a subset of $X$ consisting of the taxa for which locus $i$ provides information. $\S$ is \emph{decisive} if it satisfies the following property: Let $T$ and $T'$ be two binary phylogenetic trees for $X$ such that, for each $i \in \{1,2, \dots , k\}$, the restrictions of $T$ and $T'$ to $Y_i$ are isomorphic (restriction and isomorphism are defined in Section \ref{sec:prelims}). Then, it must be the case that $T$ and $T'$ are isomorphic.  
The \emph{decisiveness problem} is: Given a taxon coverage pattern $\S$, determine whether or not $\S$ is decisive.
Intuitively, if a taxon coverage pattern $\S$ is \emph{not} decisive, we have ambiguity.  That is, there are at least two trees that cannot be distinguished from each other by the subtrees obtained when these trees are restricted to the taxon sets in $\S$.

The complexity of the decisiveness problem has been surprisingly hard to settle, and, to our knowledge, remains an open question.  
A necessary and sufficient condition --- the four-way partition property --- for a coverage pattern to be decisive is known \cite{SteelSanderson2010,SandersonTerracesScience2011} (see also Section \ref{sec:prelims}). However, it is not clear how to test this condition efficiently.  On the positive side, the \emph{rooted} case, where at least one taxon for which every locus offers data, is known to be polynomially solvable, and software for it is available \cite{ZhbannikovDecisivator2013}.  \emph{Groves} \cite{AneEulenstein+:2009,FischerGroves2013} are a related, but not identical, notion.  For a discussion on the relationship between groves and decisiveness, see \cite{SandersonMS2010}.

\paragraph{Contributions.}
In Section \ref{sec:prelims}, we define decisiveness precisely, and review some earlier results, including the four-way partition property.  In Section \ref{sec:hyper}, we study the relationship between decisiveness and a hypergraph coloring problem --- the no-rainbow coloring problem. In Section \ref{sec:bound} we derive a lower bound on the amount of coverage needed to achieve decisiveness. Using the four-way partition property na\"ively leads to a $O(4^n)$ algorithm for decisiveness.  In Section \ref{sec:exponalg}, we show give a considerably faster exact algorithm for decisiveness.  Section \ref{sec:redFPT} studies reduction rules that allow us to compress an instance of the decisiveness problem to a smaller, but equivalent, instance.  One consequence is that we can show that the decisiveness problem is fixed-parameter tractable in the number of loci.  Section \ref{sec:ilp} presents an integer linear programming (ILP) formulation of the decisiveness problem, along with some experimental results using this formulation.  We also show that the ILP approach can be used to obtain subsets of taxa for which the given data is decisive.  Section \ref{sec:discussion} gives some concluding remarks.


\section{Preliminaries\label{sec:prelims}} 

Throughout the rest of this paper, $X$ denotes a set of taxa, $n$ denotes $|X|$, and, for any positive integer $q$, $[q]$ denotes the set $\{1, 2, \dots, q\}$.

\paragraph{Phylogenetic trees.}
A \emph{phylogenetic $X$-tree} \cite{SempleSteel03,SteelPhylogeny2016} is a tree $T$ with leaf set $X$, where every internal vertex has degree at least three.   Biologists are often interested in \emph{rooted} trees, where the root is considered as the origin of species and edges are viewed as being directed away from the root, indicating direction of evolution.  Note, however, that most phylogeny construction methods produce unrooted trees.

A \emph{split} of taxon set $X$ is a bipartition $A|B$ of $X$ such that $A, B \neq \emptyset$.  
Let $T$ be a phylogenetic $X$-tree.  Each edge $e$ in $T$ defines a split $\sigma_T(e) = A|B$, where $A$ and $B$ are the subsets of $X$ lying in each component of $T - e$.  $\Spl(T)$ denotes the set $\{\sigma_e : e \in E(T)\}$.  
It is well-known that a phylogenetic $X$-tree $T$ is completely determined by $\Spl(T )$
\cite[Theorem 3.5.2]{SempleSteel03}.  Two $X$-trees $T$ and $T'$ are \emph{isomorphic} if $\Spl(T) = \Spl(T')$.


Let $T$ be a phylogenetic $X$-tree, and suppose $Y \subseteq X$. The \emph{restriction} of $T$ to $Y$, denoted by $T |Y$, is the phylogenetic $Y$-tree where $\Spl(T | Y) = \{A \cap  Y | B \cap Y : A|B \in \Spl(T) \text{ and } A \cap Y,  B \cap Y \neq \emptyset \}.$
Equivalently, $T | Y$ is obtained from the minimal subtree of $T$ that connects $Y$ by suppressing all vertices of degree two that are not in $Y$.




\paragraph{Decisiveness.}
A taxon coverage pattern $S$ for $X$ is \emph{phylogenetically decisive} if it satisfies the following property: If $T$ and $T'$ are binary phylogenetic $X$-trees, with $T|Y = T'|Y$ for all $Y \in S$, then $T = T'$. In other words, for any binary phylogenetic $X$-tree $T$, the collection $\{T|Y : Y \in S\}$ uniquely determines $T$ (up to isomorphism).  The \emph{decisiveness problem} is the problem of determining whether a given coverage pattern is decisive.

Let $Q_S$ denote the set of all quadruples from $X$ that lie in at least one set in $S$. That is: $Q_S = \bigcup_{Y \in S} {X \choose 4}$.
A collection $S$ of subsets of $X$ satisfies the \emph{four-way partition property} (for $X$) if, for all partitions of $X$ into four disjoint, nonempty sets $A_1,A_2,A_3,A_4$ (with $A_1 \cup A_2 \cup A_3 \cup A_4 =X$) there exists $a_i \in A_i$ for $i \in \{1,2,3,4\}$ for which $\{a_1, a_2, a_3, a_4\} \in Q_S$.  

\begin{theorem}[\cite{SteelSanderson2010}]\label{thm:decisiveness}
A taxon coverage pattern $\S$ for $X$ is phylogenetically decisive if and only if $S$ satisfies the four-way partition property for $X.$
\end{theorem}


\begin{corollary}
The decisiveness problem is in co-NP. 
 \end{corollary}
 
 \begin{proof}
 A certificate for \emph{non}-decisiveness is a partition of $X$ into four disjoint, nonempty sets $A_1,A_2,A_3,A_4$, such that there is no quadruple $\{a_1,a_2,a_3,a_4\} \in Q_\S$ where $a_i \in A_i$ for each $i \in \{1,2,3,4\}$. \QED
 \end{proof}
 
 \begin{conjecture}
The decisiveness problem is co-NP-complete.
 \end{conjecture}

Note that Theorem \ref{thm:decisiveness} implies that a taxon coverage pattern $\S$ for $X$ such that $X \in \S$ (that is, one set in $\S$ contains all the taxa) is trivially decisive. 


\begin{theorem}[\cite{SteelSanderson2010}]\label{rootedDecisiveness}
Let $\S$ be a taxon coverage pattern for $X$.
\begin{enumerate}[(i)]
\item
If $\S$ is decisive, then for every set $A \in {X \choose 3}$, there exists a set $Y \in \S$ such that $A \subseteq Y$.
\item
If $\bigcap_{Y \in S} Y \neq \emptyset$, then, $\S$ is decisive if and only if for every set $A \in {X \choose 3}$, there exists a set $Y \in \S$ such that $A \subseteq Y$.
\end{enumerate}
\end{theorem}

Part (ii) of Theorem \ref{rootedDecisiveness} implies that decisiveness is polynomially solvable in the rooted case \cite{SteelPhylogeny2016}.


\section{Hypergraphs, No-Rainbow Colorings, and Decisiveness\label{sec:hyper}}

\paragraph{Hypergraphs.}
A \emph{hypergraph} 
$H$ is a pair  $H = (X,E)$, where $X$ is a set of elements called \emph{nodes} or \emph{vertices}, and 
$E$ is a set of non-empty subsets of $X$ called \emph{hyperedges} or \emph{edges} \cite{BergeGraphsHypergraphs1973,BergeHypergraphs1984}.  Two nodes $u, v \in V$ are \emph{neighbors} if $\{u,v\} \subseteq e$, for some $e \in E$.  A hypergraph $H = (X,E)$ is \emph{$r$-uniform}, for some integer $r > 0$, if each hyperedge of $H$ contains exactly $r$ nodes. 

A \emph{chain} in a hypergraph $H = (X,E)$ is an alternating sequence $v_1, e_1, v_2, \dots ,$
$e_s, v_{s+1}$ of nodes and edges of $H$ such that:
(1) $v_1, \dots , v_s$ are all distinct nodes of $H$,
(2) $e_1, \dots, e_s$ are all distinct edges of $H$, and
(3) $\{v_j,v_{j+1}\} \in e_j$ for $j \in \{1, \dots ,s\}$.
Two nodes $u, v \in X$ are \emph{connected} in $H$, denoted $u \equiv v$,  if there exists a chain in $H$ that starts at $u$ and ends at $v$. The relation $u \equiv v$ is an equivalence relation \cite{BergeGraphsHypergraphs1973}; the equivalence classes of this relation are called the \emph{connected components} of $H$.  $H$ is \emph{connected} if it has only one connected component; otherwise $H$ is \emph{disconnected}.

\paragraph{No-rainbow colorings and decisiveness.}  Let $H = (X,E)$ be a hypergraph amd $r$ be a positive integer. An \emph{$r$-coloring} of $H$ is a mapping $c : X \rightarrow [r]$.  For node $v \in X$, $c(v)$ is the \emph{color} of $v$.  Throughout this paper, $r$-colorings are assumed to be \emph{surjective}; that is, for each $i \in [r]$, there is at least one node $v \in X$ such that $c(v) = i$.
 An edge $e\in E$ is a \emph{rainbow edge} if, for each $i \in [r]$, there is at least one $v \in e$ such that $c(v) = i$.  
A \emph{no-rainbow $r$-coloring} of $H$ is a surjective $r$-coloring of $H$ such that $H$ has no rainbow edge.

Given an $r$-uniform hypergraph $H = (X,E)$, the \emph{no-rainbow $r$-coloring problem} ($r$-NRC) asks whether $H$ has a no-rainbow $r$-coloring \cite{bodirsky2012complexity}.  $r$-NRC is clearly in NP, but it is unknown whether the problem is NP-complete \cite{bodirsky2012complexity}.  

Let $\S$ be a taxon coverage pattern for $X$.  We associate with $\S$ a hypergraph $H(\S) = (X, \S)$, and with $Q_\S$, we associate a $4$-uniform hypergraph $H(Q_\S) = (X, Q_\S)$. 
The next result states that $r$-NRC is equivalent to the complement of the decisiveness problem.  

\begin{proposition}\label{lem:decRain}
Let $\S$ be a taxon coverage pattern.  The following statements are equivalent.
\begin{enumerate}[(i)]
\item
\S\ is not decisive.
\item
$H(Q_\S)$ admits a  no-rainbow $4$-coloring.
\item
$H(\S)$ admits a no-rainbow $4$-coloring.
\end{enumerate}
\end{proposition}

\begin{proof}
(1) $\Leftrightarrow$ (2): By Theorem \ref{thm:decisiveness}, it suffices to show that $\S$ fails to satisfy the $4$-way partition property if and only if $H(Q_\S)$ has a no-rainbow $4$-coloring. $\S$ does not satisfy the $4$-way partition property if and only if there exists a $4$-way partition $A_1, A_2, A_3, A_4$ of $X$ such that, for every $q \in Q_\S$, there is an $i \in [4]$ such that $A_i \cap q = \emptyset$.  This holds if and only if the coloring $c$, where $c(v) = i$ if and only if $v \in A_i$, is a no-rainbow $4$-coloring of  $H(Q_S)$.  

(2) $\Leftrightarrow$ (3):  It can be seen that if $c$ is a no-rainbow 4-coloring of $H(\S)$, then $c$ is a no-rainbow 4-coloring of $H(Q_S)$.  We now argue that if $c$ is a no-rainbow 4-coloring of $H(Q_\S)$, then $c$ is a no-rainbow 4-coloring of $H(\S)$.  Suppose, to the contrary, that there a rainbow edge $Y \in \S$.  Let $q$ be any $4$-tuple $\{v_1, v_2, v_3, v_4\} \subseteq Y$ such that $c(v_i) = i$, for each $i \in [4]$.  Then, $q$ is a rainbow edge in $Q_\S$, a contradiction. \QED
\end{proof}


\begin{proposition}\label{prop:ccs}
Let $H = (X,E)$ be a hypergraph and $r$ be a positive integer.  
\begin{enumerate}[(i)]
\item
If $H$ has at least $r$ connected components, then $H$ admits a no-rainbow $r$-coloring.
\item
If $r = 2$, then $H$ admits a no-rainbow $r$-coloring if and only if $H$ is disconnected.
\end{enumerate}
\end{proposition}

 \begin{proof}
(i) Suppose the connected components of $H$ are $C_1, \dots , C_q$, where $q \ge r$.  For each $i \in \{1, \dots , r-1\}$, assign color $i$ to all nodes in $C_i$.  For $i = \{r, \dots , q\}$, assign color $r$ to all nodes in $C_i$. Thus, no edge is rainbow-colored. 

(ii) By part (i), if $H$ is disconnected, it admits a no-rainbow 2-coloring.  To prove the other direction, assume, for contradiction that $H$ admits a no-rainbow 2-coloring but it is connected. Pick any two nodes $u$ and $v$ such that $c(u) = 1$ and $c(v) =2$.  Since $H$ is connected, there is a $(u,v)$-chain in $H$.  But this chain must contain an edge with nodes of two different colors; i.e., a rainbow edge.
\QED
 \end{proof} 

Part (ii) of Proposition \ref{prop:ccs} implies the following.
 
 \begin{corollary}\label{cor:2RNC}
 $2\text{-NRC}\in P$.
 \end{corollary}

\begin{lemma}\label{lem:nonneighbors} 
Let $H = (X,E)$ be an $r$-uniform hypergraph.
Suppose that there exists a subset $A$ of $X$ such that $2 \le |A| \le r-1$ and $A \not\subseteq e$ for any $e \in E$.  Then, $H$ has a no-rainbow $r$-coloring. 
\end{lemma}

\begin{proof}
Let $c$ be the coloring where each of the nodes in $A$ is assigned a distinct color from the set $[|A|]$ and the remaining nodes are assigned colors from the set $\{|A|+1, \dots ,r \}$. Then, $c$ is a no-rainbow $r$-coloring of $H$. \QED
\end{proof}



\section{A Tight Lower Bound on the Coverage\label{sec:bound}}


The next result provides a tight lower bound on the minimum amount of coverage that is needed to achieve decisiveness. 

\begin{theorem}\label{thm:lowerBoundS}
Let $\S$ be a taxon coverage pattern for $X$ and let $n = |X|$. If $\S$ is decisive, then $|Q_\S| \ge {n- 1 \choose 3}$.  This lower bound is tight. That is, for each $n \ge 4$, there exists a decisive  taxon coverage pattern $\S$ for $X$ such that $|Q_\S| = {n- 1 \choose 3}$.
\end{theorem}

To prove Theorem \ref{thm:lowerBoundS}, 
for every pair of integers $n, r$ such that $n \ge r \ge 1$
let us define the function 
$A(n,r)$ as follows.
$$A(n,r) = 
\begin{cases}
1 & \text{if $r = 1$ or $n = r$}\\
A(n-1,r-1) + A(n-1, r) & \text{otherwise.}
\end{cases}$$

\begin{lemma}\label{thm:lowerBound}
Let $n$ and $r$ be integers such that $n \ge r \ge 1$ and let $H = (X,E)$ be an $n$-vertex $r$-uniform hypergraph.  If $|E| < A(n,r)$, then $H$ admits a no-rainbow $r$-coloring.  If $|E| \ge A(n,r)$,  then $H$ may or may not admit a no-rainbow $r$-coloring.  Furthermore, there exist $n$-vertex $r$-uniform hypergraphs with exactly $A(n,r)$ edges that do not admit a no-rainbow $r$-coloring.
\end{lemma}

\begin{proof} 
For $r=1$ or $n=r$, $H$ has at most one hyperedge.  If $H$ has exactly one hyperedge, then any coloring that uses all $r$ colors contains a rainbow edge.  If $H$ contains no hyperedges, then $H$ trivially admits a no-rainbow $r$-coloring.

Let us assume that  for any $i$ and $j$ with $1\leq i < n$ and $1\leq j \leq r$, $A(i,j)$ equals the minimum number of hyperedges an $i$-node, $j$-uniform hypergraph $H$ that does not admit a no-rainbow $r$-coloring.  We now prove the claim for $i = n$ and $j =r$.

Pick an arbitrary node $v \in X$. There are two mutually disjoint classes of  colorings of $H$:
(1) the colorings $c$ such that $c(v) \neq c(u)$ for any $u \in X \setminus \{v\}$, and
(2) the colorings $c$ such that $c(v) = c(u)$ for some $u \in X \setminus \{v\}$.

For the colorings in class 1, we need hyperedges that contain node $v$, since in the absence of such hyperedges, any coloring is a no-rainbow coloring. Assume, without loss of generality, that $c(v) = r$. The question reduces to finding the number of hyperedges in an $(n-1)$-node $(r-1)$-uniform hypergraph (since $v$'s color, $r$, is unique).  
The minimum number of hyperedges needed to avoid a no-rainbow $(r-1)$-coloring for an $(n-1)$-node hypergraph is $A(n-1,r-1)$. 

To find the minimum number of hyperedges needed to cover colorings of class 2, we ignore $v$, since $v$ is assigned a color that is used by other nodes as well.  
The number of hyperedges needed for this class is $A(n-1,r)$.

To obtain a lower bound, we add the lower obtained for the two disjoint classes of colorings. Thus, $A(n,r) = A(n-1,r-1) + A(n-1, r)$.
\QED
\end{proof}

\begin{lemma}\label{lem:AandC}
$A(n,r) = {n-1 \choose r-1}$.
\end{lemma}

\begin{proof}
For $r=1$, $A(n,r) = {n-1 \choose 0} = 1$ and for $n=r$, $A(n,r) = A(r,r) = {r-1 \choose r-1} = 1$. 
Now, assume that  $A(i,j) = {i-1 \choose j-1}$, for $1\leq i\leq n-1$ and $1\leq j\leq r$. Then, $A(n,r) = A(n-1,r) + A(n-1,r-1) = {n-2 \choose r-1} + {n-2 \choose r-2} = {n-1 \choose r-1}$. 
\QED
\end{proof}

\begin{proof}[of Theorem \ref{thm:lowerBoundS}]
Follows from Lemmas \ref{thm:lowerBound} and \ref{lem:AandC}, by setting $r = 4$, and Proposition \ref{lem:decRain}(ii). \QED
\end{proof}
\section{An Exact Algorithm for Decisiveness\label{sec:exponalg}}

The na\"ive way to use Theorem \ref{thm:decisiveness} to test whether a coverage pattern $\S$ is decisive is to enumerate all partitions of $X$ into four non-empty sets $A_1, A_2, A_3, A_4$ and verify that there is a set $Y \in \S$ that intersects each $A_i$.  Equivalently, by Proposition \ref{lem:decRain}, we can enumerate all surjective colorings of $H(S)$ and check if each of these colorings yields a rainbow edge.  In either approach, the number of options to consider is given by a Stirling number of the second kind, namely $\left\{{n \atop 4}\right\} \sim \frac{4^n}{4!}$ \cite{GrahamKnuthPatashnik1989}.  The next result is a substantial improvement over the na\"ive approach.

\begin{theorem}\label{thm:exactExpon}
Let $\S$ be a taxon coverage pattern for a taxon set $X$.  Then, there is an algorithm that, in $O^*(2.8^n)$ time\footnote{The $O^*$-notation is a variant of $O$-notation that ignores polynomial factors \cite{FominKratsch2010}.} determines whether or not $\S$ is decisive.
\end{theorem}

The proof of Theorem \ref{thm:exactExpon} relies on the following result.

\begin{algorithm}[t]
\Fn(){\FindNRC{$H$}}{
\SetAlgoLined
\SetNoFillComment
\DontPrintSemicolon
\KwIn{A 4-uniform hypergraph $H = (X,E)$ such that $|X| \ge 4$.}
\KwOut{A no-rainbow $4$-coloring of $H$, if one exists; otherwise, \texttt{fail}.}

\For {$i = 1$  \KwTo $\lfloor \frac{n}{4}\rfloor$}{
  \lForEach{$v \in X$}{$c(v) = \texttt{uncolored}$}
  \ForEach {$A \subseteq X$ such that $|A| = i$}{ 
    \lForEach{$v \in A$}{$c(v) = 1$}

    \For {$j = 1$ \KwTo $\lfloor \frac{n-i}{3} \rfloor$}{
      \ForEach {$B \subseteq X \setminus A$ such that $|B| = j$}{
      	\lForEach{$v \in B$}{$c(v) = 2$}
        \If{there is no $e \in E$ such that, for each $i \in [2]$, $m_e^c(i) = 1$}
          { Arbitrarily split $X \setminus (A \cup B)$ into nonempty sets $C, D$ \; 
          \lForEach{$v \in C$}{$c(v) = 3$}
          \lForEach{$v \in D$}{$c(v) = 4$}
          \Return $c$
          }
          \Else{
          	Choose any $e \in E$ such that $m_e^c(i) = 1$ for each $i \in [2]$ \;
		\lForEach{uncolored node $x \in e$}{$c(x) = 3$}
		\While{there exists $e \in E$ s.t.\ $m_e^c(i) = 1$ for each $i \in [3]$}{ 
            		Pick any $e  \in E$ s.t.\ $m_e^c(i) = 1$ for each $i \in [3]$ \;
			Let $x$ be the unique uncolored node in $e$ \;
            		$c(x) = 3$
          }
        \lIf{$X$ contains no uncolored node}{
          \Return \texttt{fail}
        }
        \Else{
          \lForEach{uncolored vertex $u \in X$}{$c(u) = 4$}
          \Return $c$
        }

	}
    }
  }
  }
}
        		\Return \texttt{fail}

}
\smallskip
\caption{No-rainbow 4-coloring of $H$.
\label{alg:nonRainbow}}
\end{algorithm}

\begin{lemma}\label{lem:exactExpon}
There exists an algorithm that, given a 4-uniform hypergraph $H = (X,E)$, determines if $H$ has a no-rainbow $4$-coloring in time $O^*(2.8^n)$.
\end{lemma}

\begin{proof}
We claim that algorithm \FindNRC (Algorithm \ref{alg:nonRainbow}) solves 4-NRC in $O^*(2.8^n)$ time.  \FindNRC relies on the observation that if $H$ has a no-rainbow 4-coloring $c$, then (1) there must exist a subset $A \subseteq X$ where $|A| \le \lfloor \frac{n}{4} \rfloor$, such that all nodes in $A$ have the same color, which is different from the colors used for $X \setminus A$, and (ii) there must exist a subset $B \subseteq X \setminus A$, where $|B| \le \lfloor \frac{n - |A|}{3} \rfloor$, such that all nodes in $B$ have the same color, which is different from the colors used for the nodes in $X \setminus B$.   \FindNRC tries all possible choices of $A$ and $B$ and, without loss of generality, assigns $c(v) = 1$, for all $v \in A$ and $c(v) = 2$, for all $v \in B$.  We are now left with the problem of determining whether we can assign colors 3 and 4 to the nodes in $X \setminus (A \cup B)$ to obtain a no-rainbow 4-coloring for $H$.

Let $c$ be the current coloring of $H$.  For each $e \in E$ and each $i \in [4]$, $m_e^c(i)$ denotes the number of nodes $v \in e$ such that $c(v) = i$.  Consider the situation after \FindNRC assigns colors 1 and 2 to the nodes in $A$ and $B$.  There are two cases, both of which can be handled in polynomial time.
\begin{enumerate}
\item  \emph{There is no $e \in E$, such that, for each $i \in [2]$, $m_e^c = 1$.}  Then, if we partition the nodes of $X \setminus (A \cup B)$, arbitrarily into subsets $C$ and $D$ and assign $c(v) = 3$ for each $v \in C$ and $c(v) =4$ for each $v \in D$, we obtain a no-rainbow 4-coloring of $H$.
\item
 \emph{There exists $e \in E$, such that, for each $i \in [2]$, $m_e^c = 1$.} Let $e$ be any such edge.  Then $e$ must exactly contain two uncolored nodes, $x$ and $y$.  To avoid $e$ becoming a rainbow edge, we must set $c(x) = c(y) \not\in [2]$.  Without loss of generality, make $c(x) = c(y) = 3$. Next, as long as there exists any hyperedge $e$ such that $m_e^c(i) =1$ for each $i \in [3]$, the (unique) uncolored node $x$ in $e$ must be assigned $c(x) = 3$, because setting $c(x) = 4$ would make $e$ a rainbow hyperedge. 
Once no such hyperedges remain, we have two possibilities:
\begin{enumerate}
\item  \emph{$X$ does not contain uncolored nodes.}  Then, there does not exist a no-rainbow 4-coloring, given the current choice of $A$ and $B$.
\item \emph{$X$ contains uncolored nodes.} Then, there is no $e \in E$ such that $m_e^c(i) = 1$ for each $i \in [3]$.  Thus, if we set $c(u) = 4$ for each uncolored node $u$, we obtain a no-rainbow 4-coloring for $H$.
\end{enumerate}
\end{enumerate}

The total number of pairs $(A,B)$ considered throughout the execution of \FindNRC is at most $\sum_{i=1}^{\lfloor \frac{n}{4} \rfloor} {n \choose i} \sum_{j=1}^{\lfloor \frac{n-i}{3} \rfloor} {n-i \choose j}$.  We have estimated this sum numerically to be $O(2.8^n)$.  The time spent per pair $(A,B)$ is polynomial in $n$; hence, the total running time of \FindNRC is $O^*(2.8^n)$.
\QED
\end{proof}

\begin{proof}[of Theorem \ref{thm:exactExpon}]
Given $\S$, we construct the hypergraph $H(Q_\S)$, which takes time polynomial in $n$, and then run \FindNRC{$H(Q_\S)$}, which, by Lemma \ref{lem:exactExpon}, takes $O^*(2.8^n)$ time.  If the algorithm returns a no-rainbow 4-coloring $c$ of $H(Q_\S)$, then, by Proposition \ref{lem:decRain}, $\S$ is not decisive; if  \FindNRC{$H(Q_\S)$} returns \texttt{fail}, then $\S$ is decisive.
\QED
\end{proof}

\section{Reduction Rules and Fixed Parameter Tractability\label{sec:redFPT}}
 
A \emph{reduction rule} for the decisiveness problem is a rule that replaces an instance $\S$ of the problem by a smaller instance $\Sred$ such that $\S$ is decisive if and only if $\Sred$ is.  Here we present reduction rules that can reduce an instance of the decisiveness problem into a  one whose size depends only on $k$.  This size reduction is especially significant for taxon coverage patterns where the number of loci, $k$, is small relative to the number of taxa.  Such inputs are not uncommon in the literature --- examples of such data sets are studied in Section \ref{sec:ilp}.
  
We need to introduce some definitions and notation.  Let $H = (X,E)$ be a hypergraph where $X = \{x_1, x_2, \dots , x_n\}$ and $E = \{e_1, e_2, \dots, e_k\}$.  The \emph{incidence matrix} of $H$ is the $n \times k$ binary matrix where $M_H[i,j] = 1$ if $x_i \in e_j$ and $M_H[i,j] = 0$ otherwise.
Two rows in $M_H$ are \emph{copies} if the rows are identical when viewed as $0$-$1$ strings; otherwise, they are \emph{distinct}. 

Let $\Mred_H$ denote the matrix obtained from $M_H$ by striking out duplicate rows, so that $\Mred_H$ retains only one copy of each row in $M_H$.  Let $\nred$ denote the number of rows of $\Mred_H$.  Then, $\nred \le 2^k$.  $\Mred$ is the incidence matrix of a hypergraph $\Hred = (\Xred,\Ered)$, where $\Xred \subseteq X$, and each $v \in \Xred$ corresponds to a distinct row of $M_H$.  For each $v \in X$, $X(v) \subseteq X$ consists of all nodes $u \in X$ that correspond to copies of the row of $M_H$ corresponding to $v$.

%

Given two binary strings $s_1$ and $s_2$ of length $k$, $s_1 \, \& \, s_2$ denotes the bitwise \emph{and} of $s_1$ and $s_2$; $\mathbf{0}$ denotes the all-zeroes string of length $k$.

The next result is a direct consequence of Lemma \ref{lem:nonneighbors}.

\begin{proposition}\label{prop:AND}
If $\Mred_H$ has two rows $r_1$ and $r_2$ such that $r_1 \, \& \, r_2 = \mathbf{0}$ or three rows $r_1$, $r_2$ and $r_3$ such that $r_1 \, \& \, r_2 \, \& \, r_3 = \mathbf{0}$, then $\Hred$ and $H$ admit no-rainbow 4-colorings.
\end{proposition}

\begin{corollary}\label{cor:2k}
If $\Mred_H$ has more than $2^{k-1}$ rows, where $k$ is the number of columns, then $H$ admits a no-rainbow 4-coloring.
 \end{corollary}

\begin{proof}
Suppose $\nred \ge 2^{k-1}$. Then, there are at least two rows $r_1$ and $r_2$ in $\Mred_H$ that are complements of each other (that is, $r_2$ is is obtained by negating each bit in $r_1$) and, thus, $r_1 \, \& \, r_2 = \mathbf{0}$. The claim now follows from Proposition \ref{prop:AND}.
\QED
\end{proof}

\begin{theorem}\label{thm:fptThm}
Suppose $n \geq \nred+2$.  $H$ admits a no-rainbow 4-coloring if and only if $\Hred$ admits a no-rainbow $r$-coloring for some $r \in \{2,3,4\}$.
\end{theorem}
 
\begin{proof}
(\emph{If}) 
Suppose $\Hred = (\Xred, \Ered)$ admits a no-rainbow $r$-coloring $\cred$ for some $r \in \{2,3,4\}$.  Let $c$ be the coloring for $H$ obtained by setting $c(u) = \cred(v)$, for each $v \in \Xred$ and each $u \in X(v)$.  If $\cred$ is a no-rainbow 4-coloring of $\Hred$, then $c$ is  also one for $H$, and we are done.  Suppose $\cred$ is a 3-coloring.  Since $n \geq \nred+2$, there must exist $v \in \Xred$ such that $|X(v)| \ge 2$.  We choose one node $u \in X(v) \setminus \{v\}$, and set $c(u) = 4$, making $c$ a no-rainbow 4-coloring for $H$.  Suppose $\cred$ is a no-rainbow 2-coloring. If there exists $v \in \Xred \setminus\{v\}$ such that $|X(v)| \ge 3$, we pick pick any $u, w \in \Xred \setminus\{v\}$, and set $c(u) = 3$ and $c(w) = 4$. If there is no $v \in \Xred$ such that $|X(v)| \ge 3$, there must exist $v_1 , v_2 \in \Xred$ such that $|X(v_i)| \ge 2$ for $i \in \{1,2\}$.  For $i \in \{1,2\}$, choose any $u_i \in X(v_i) \setminus \{v_i\}$ and set $c(u_i) = i+2$.

(\emph{Only if})
Suppose $H$ has a no-rainbow 4-coloring $c$. Let $\cred$ be the coloring of $\Hred$ where, for each $v \in \Xred$, $\cred(v) = c(u)$, for some arbitrarily chosen node in $u \in X(v)$.  If $\cred$ is a surjective 4-coloring of $v$, $\cred$ must be a no-rainbow 4-coloring of $\Hred$, and we are done.  In Appendix \ref{appA}, we show that any $r$-coloring of $\Hred$, where $r \in \{1,2,3,4\}$ can be converted into a no-rainbow 4-coloring of $\Hred$ by altering some of the colors assigned by $\cred$.
\QED
 \end{proof}
 
\begin{theorem}\label{thm:fpt}
 Decisiveness is fixed-parameter tractable in $k$.
\end{theorem}
 
\begin{proof}
Let $\S$ be the input coverage pattern.  First, in  $O^*(2^k)$ time, we construct $\Hred(S)$.
 By Theorem \ref{rootedDecisiveness}, we need to test if $\Hred(\S)$ admits a non-rainbow $r$-coloring for any $r \in \{2,3,4\}$.  If the answer is ``yes'' for any such $r$, then $\S$ is not decisive; otherwise $\S$ is decisive.   By Corollary \ref{cor:2RNC}, the test for $r =2$ takes polynomial time. We perform the steps for $r =3$ and $r=4$ using the algorithm of Section \ref{sec:exponalg}.  The total time is $O^*(2.8^{\nred})$, which is $O^*(2.8^{2^{k}})$. \QED
\end{proof}
 


\section{An Integer Linear Programming Formulation\label{sec:ilp}}

Let $\S = \{Y_1, Y_2, \dots , Y_k\}$ be a taxon coverage pattern for $X$.  Here we formulate a 0-1 integer linear program (ILP) that is feasible if and only if \S\ is non-decisive\footnote{For an introduction to the applications of integer linear programming, see \cite{Gusfield_IPinCB2018}.}.  We use the equivalence between non-decisiveness of \S\ and the existence of a no-rainbow $4$-coloring of hypergraph $H(\S)$ (Proposition \ref{lem:decRain}).

Suppose $X = \{a_1, a_2, \dots, a_n\}$.  For each $i \in [n]$ and each color $q \in [4]$, define a binary \emph{color variable} $x_{iq}$, where $x_{iq} = 1$ if taxon $i$ is assigned color $q$. To ensure that each $i \in X$ is assigned only one color, we add constraints 
\begin{equation}\label{eqn:colorConstr}
\sum_{q \in [4]} x_{iq} = 1,\quad \text{for each $i \in X$}.
\end{equation}

The following constraints ensure that each color $q \in [4]$ appears at least once.
\begin{equation}\label{eqn:atLeastConstr}
\sum_{i \in X} x_{iq} \ge 1,\quad \text{for each $q \in [4]$}.
\end{equation}

To ensure that, for each $j \in [k]$, $Y_j$ is not rainbow colored, we require that there exist at least one color that is not used in $Y_j$; i.e, that $\sum_{i \in Y_j} x_{iq} = 0$, for some $q \in [4]$.  To express this condition,
for each $j \in [k]$ and each $q \in [4]$, we define a binary variable $z_{jq}$, which is $1$ if and only if $\sum_{i \in Y_j} x_{iq} = 0$.  We express $z_{jq}$ using the following linear constraints.
\begin{equation}\label{eqn:orRel}
 (1 - z_{jq}) \le \sum_{i \in Y_j} x_{iq} \le n \cdot (1-z_{jq}),\quad \text{for each $j \in [k]$ and each $q \in [4]$}
\end{equation}

The requirement that $Y_j$ not be rainbow-colored is expressed as
\begin{equation}\label{eqn:atLeastOneZeroBool}
\sum_{q \in [4]} z_{jq} \ge 1, \quad \text{for each $j \in [k]$.}
\end{equation}

\begin{proposition}\label{prop:ILP}
$\S$ is non-decisive if and only if the 0-1 ILP with variables $x_{iq}$ and $z_{jq}$ and constraints  \eqref{eqn:colorConstr}, \eqref{eqn:atLeastConstr}, \eqref{eqn:orRel}, and \eqref{eqn:atLeastOneZeroBool} is feasible.
\end{proposition}


%
%

\paragraph{Experimental Results.}
Here we summarize our computational results using ILP on the data sets studied in \cite{DobrinZwicklSanderson2018}.  For details, see Appendix \ref{app:ILP}.

We generated ILP models for all the data sets in  \cite{DobrinZwicklSanderson2018} and
used Gurobi \cite{gurobi} to solve 9 of these models. All but one of these models were solved in under 0.1 seconds.    
Only one of the data sets, Insects, is decisive (and its ILP took the longest to solve).  Indeed, the Insects data set is trivially decisive, as one locus spans all the taxa.  For the remaining data sets, we used a simple heuristic to identify a subset of the taxa for which the data is decisive. If the data set is non-decisive, we remove the taxon covered by the fewest loci, breaking ties in favor of the first taxon in the input.  After removing a taxon, we update the ILP model and run it again. 
When the model becomes infeasible, the remaining data set is decisive. 

For two data sets (Saxifragales and Mammals), the heuristic yielded trivially decisive coverage patterns. 
We obtained non-trivial results for three data sets.
For the complete Birds data set, the largest of all, the heuristic took 1.1 hours. 
Although the heuristic retained only 2.5\% of the original taxa, every family of taxa except one from the original data set 
is represented in the final result. For Bats, the heuristic took 70 seconds and achieved 4.3\% coverage, but had sparse coverage across the families.  For Primates, the heuristic took 33 seconds and achieved 50.3\% coverage, distributed over most families. For all data sets, the most time-consuming step was attempting to solve the final, infeasible, ILP.

%


\section{Discussion\label{sec:discussion}}

Despite its apparent complexity, the decisiveness problem appears to be quite tractable in practice.  Since real data sets are likely to be non-decisive, testing for decisiveness can only be considered a first step.  Indeed, if we determine that a data set is not decisive, it is useful to find a subset of the data that is decisive.  In Section \ref{sec:ilp}, we have taken some preliminary steps in that direction, using a simple heuristic. This heuristic could potentially be improved upon, perhaps relying on the data reduction ideas of Section \ref{sec:redFPT}. One open problem is whether the doubly-exponential algorithm of Theorem \ref{thm:fpt} can be improved.

\paragraph{Acknowledgements.}  M.\ Steel pointed out the connection between decisiveness and hypergraph coloring. We thank M.\ Sanderson for useful discussions.  He and B.\ Dobrin provided the data studied in section \ref{sec:ilp}.

\appendix

\newpage

\section{Proof of the ``Only If'' Direction of Theorem \ref{thm:fptThm}\label{appA}}

Here we show that if $H = (X,E)$ admits a no-rainbow 4-coloring $c$, then $\Hred = (\Xred, \Ered)$ admits a no-rainbow $r$-coloring for some $r \in \{2,3,4\}$.  

Let $\cred$ be the coloring of $\Hred$ where, for each $v \in \Xred$, $\cred(v) = c(v)$.

\begin{lemma}\label{lem:3color}
Suppose \cred\ is a 3-coloring of $\Hred$ that has a rainbow edge.  Then, we can transform $\cred$ into a no-rainbow 4-coloring $\widehat{c}$ of $\Hred$.
\end{lemma}

\begin{proof}
Consider three nodes $v_1,v_2,v_3 \in \Xred$ that are all neighbors and have distinct colors. Since $\Xred \subseteq X$, $v_1,v_2, v_3$ are neighbors in $H$ as well. Since $H$ admits a no-rainbow 4-coloring, there must exist a node $v_4 \in X \setminus \{v_1,v_2, v_3\}$ such that $c(v_4) \neq c(v_i)$, for $i\in\{1,2,3\}$. Since $v_4 \notin \Xred$ there must be a copy of $v_4$ with a different color in $\Xred$. If we insert $v_4$ to $\Xred$ and remove its copy from $\Xred$, then $\Xred$ has 4 colors and the new coloring is a no-rainbow 4-coloring. 
\QED
\end{proof}

\begin{lemma}\label{lem:2color}
Suppose \cred\ is a 2-coloring of $\Hred$, that has a rainbow edge.  Then, we can transform $\cred$ into a no-rainbow $r$-coloring $\widehat{c}$ of $\Hred$, for $r\in\{3,4\}$. 
\end{lemma}

\begin{proof}
Consider any two neighbor nodes $v_1,v_2 \in \Xred$ such that $v_1$ and $v_2$ have distinct colors. Note that $v_1$ and $v_2$ must be neighbors in $H$ as well. Since $H$ admits no-rainbow 4-coloring, there must be at least one other node $v_3 \in X$ such that $v_3 \neq v_i$, for  $i\in\{1,2\}$ and $c(v_3) \neq c(v_i)$, for $i\in\{1,2\}$. Since $v_3 \notin \Xred$ there must be a copy of $v_3$ with another color in $\Hred$. If we insert $v_3$ with color $c(v_3)$ to $\Hred$ and remove its copy from $\Hred$, then $\Hred$ has 3 colors. If the new coloring is a no-rainbow 3-coloring we are done, otherwise the new coloring is a 3-coloring with a rainbow edge. In the latter case, we use Lemma \ref{lem:3color} to find a no-rainbow 4-coloring. \QED
\end{proof}

\begin{lemma}\label{lem:1color}
Suppose \cred\ is a 1-coloring of $\Hred$.  Then, we can transform $\cred$ into a no-rainbow $r$-coloring $\widehat{c}$ of \Hred, for $r\in\{2,3,4\}$.
\end{lemma}

\begin{proof}
Since $H$ admits no-rainbow 4-coloring, we can replace one of the nodes of $\Hred$ with a copy of that in $H$ of different color. The result is a 2-coloring for $\Hred$. If the result is a no-rainbow 2-coloring we are done, otherwise we have a 2-coloring of $\Hred$ with a rainbow edge. In the latter case, we use Lemma \ref{lem:2color} to find a no-rainbow $r$-coloring for $r\in \{3,4\}$.
\QED
\end{proof}

%

To prove of the ``only if'' direction of Theorem \ref{thm:fptThm}, note first that if \cred\ is a no-rainbow $r$-coloring of $\Hred$, for $r \in \{2,3,4\}$, we are done.  Notice that if $\cred$ is a 4-coloring of $\Hred$, it must be a no-rainbow coloring of \Hred. Otherwise, $\cred$ is an $r$-coloring of $\Hred$, for $r \in \{1,2,3\}$, that has a rainbow edge.  In this case, we apply Lemmas \ref{lem:3color}, \ref{lem:2color}, \ref{lem:1color}, as appropriate, to obtain a no-rainbow $r$-coloring of $\Hred$, for some $r \in \{2,3,4\}$.
\QED

\newpage

\section{Computational Results Using ILP\label{app:ILP}}

We wrote a Python script that given a taxon coverage pattern, generates an ILP model as described in Section \ref{sec:ilp}. Table \ref{tab:ILP} shows the time taken to generate the ILP models for the data sets analyzed in  \cite{DobrinZwicklSanderson2018} (see the latter reference for full citations of the corresponding phylogenetic studies). The models were generated on a Linux server.
Table \ref{tab:ILPsizes} shows the sizes of several of these ILPs.  Table \ref{tab:ILPtimes} shows the time taken by Gurobi to solve each of the latter ILPs, on a Lenovo Thinkpad X1 Carbon running Windows.

\begin{table}
\caption{Running times for generating ILPs for data sets studied in \cite{DobrinZwicklSanderson2018}.}\label{tab:ILP}
\begin{center}
\begin{tabular}{|l|c|c|c|}
 \hline
& ~Execution Time~ & ~Number of~  & ~Number of~  \\
 Data Set & (seconds) & Taxa & Loci \\ 
\hline
Allium&0.051037&57&6\\
Asplenium&0.047774&133&6\\
Bats&0.152805&815&29\\
Birds (complete)&4.688950&7000&32\\
Birds &2.723334&5146&32\\
Caryophyllaceae&0.068084&224&7\\
Chameleons&0.059073&202&6\\
Eucalyptus&0.058591&136&6\\
Euphorbia&0.061188&131&7\\
Ficus&0.063072&112&5\\
Fungi&0.223971&1317&9\\
Insects&7.649374&144&479\\
Iris&0.055743&137&6\\
Mammals&0.110263&169&26\\
Primates&0.363623&372&79\\
Primula&0.064607&185&6\\
Scincids&0.071276&213&6\\
Ranunculus&0.059699&170&7\\
Rhododendron&0.052903&117&7\\
Rosaceae&0.092148&529&7\\
Solanum&0.062660&187&7\\
Saxifragales&0.173522&946&51\\
Szygium&0.051021&106&5\\
 \hline
\end{tabular}
\end{center}
\end{table}

\begin{table}
\caption{Sizes of the ILPs for a subset of the data in \cite{DobrinZwicklSanderson2018}.}
\label{tab:ILPsizes}
\begin{center}
\begin{tabular}{|l|c|c|c|} \hline
& ~Number of~ & ~Number of~ & ~Number of~ \\
Data Set & Rows & Columns & ~Nonzero Entries~ \\ \hline
Bats & 1080 & 3376 & 34884 \\
Birds (complete) & 7292 & 28128 & 273416 \\
Eucalyptus & 194 & 568 & 2680 \\
Ficus & 161 & 468 & 2548 \\
Insects &  4459 & 2492 & 529140 \\
Iris & 195 & 572 & 3320 \\
Mammals & 407 & 780 & 34688 \\ 
Primates & 1087 & 1804 & 91356 \\
Saxifragales & 1409 & 3988 & 30060 \\ \hline
\end{tabular}
\end{center}
\label{default}
\end{table}%

\begin{table}
\caption{Solution times for the ILPs listed in Table \ref{tab:ILPsizes}.}
\label{tab:ILPtimes}
\begin{center}
\begin{tabular}{|l|c|} \hline
& ~Execution Time~  \\
 Data Set & (seconds)  \\ 
\hline
Bats & 0.098 \\
Birds (complete) & 0.03 \\
Eucalyptus & 0.002 \\
Ficus & 0.0009999 \\
Insects & 5.902 \\
Iris & 0.002 \\
Mammals & 0.091 \\
Primates & 0.013 \\
Saxifragales & 0.004 \\ \hline
\end{tabular}
\end{center}
\label{default}
\end{table}%

To produce decisive submatrices, multiple ILPs were solved. The average time taken per ILP for some of the data sets is as follows:
\begin{itemize}
\item Birds (complete): 0.59 seconds.
\item Bats: 0.09 seconds.
\item Primates: 0.18 seconds.
\end{itemize}

\end{document}